\documentclass[10pt, conference, letterpaper]{IEEEtran}
\IEEEoverridecommandlockouts
\usepackage{cite}
\usepackage{amsmath,amssymb,amsfonts,amsthm}
\usepackage{algorithmic}
\usepackage{graphicx}
\usepackage{textcomp}
\usepackage{xcolor}

\usepackage{booktabs}
\usepackage{enumerate}
\usepackage{makecell}
\usepackage{url}

\usepackage{multirow,float}
\usepackage{bm}   
\usepackage[linesnumbered,ruled,vlined]{algorithm2e}
\usepackage{xcolor}
\usepackage{pifont}

\SetCommentSty{mycommfont}
\usepackage{tikz,pgfplots}
\usetikzlibrary{pgfplots.groupplots}
\usetikzlibrary{patterns}
\usetikzlibrary{shapes,snakes}
\usetikzlibrary{matrix}
\usepgfplotslibrary{fillbetween}

\newtheorem{theorem}{Theorem}[section]

\theoremstyle{definition}
\newtheorem{definition}{Definition}[section]
\newtheorem*{remark}{Remark}

\usetikzlibrary{external}
\begin{document}

\title{HTNet: Dynamic WLAN Performance Prediction using Heterogenous Temporal GNN
}

\author{
    \IEEEauthorblockN{Hongkuan Zhou\IEEEauthorrefmark{1}, Rajgopal Kannan\IEEEauthorrefmark{2}, Ananthram Swami\IEEEauthorrefmark{2},Viktor Prasanna\IEEEauthorrefmark{1}}
    \IEEEauthorblockA{\IEEEauthorrefmark{1}University of Southern California}
    \IEEEauthorblockA{\{hongkuaz,prasanna\}@usc.edu}
    \IEEEauthorblockA{\IEEEauthorrefmark{2}DEVCOM US Army Research Lab}
    \IEEEauthorblockA{\{rajgopal.kannan.civ,ananthram.swami.civ\}@army.mil}
}

\maketitle

\begin{abstract}
Predicting the throughput of WLAN deployments is a classic problem that occurs in the design of robust and high performance WLAN systems. However, due to the increasingly complex communication protocols and the increase in interference between devices in denser and denser WLAN deployments, traditional methods either have substantial runtime or enormous prediction error and hence cannot be applied in downstream tasks. Recently, Graph Neural Networks have been proven to be powerful graph analytic models and have been broadly applied to various networking problems such as link scheduling and power allocation. In this work, we propose HTNet, a specialized Heterogeneous Temporal Graph Neural Network that extracts features from dynamic WLAN deployments. Analyzing the unique graph structure of WLAN deployment graphs, we show that HTNet achieves the maximum expressive power on each snapshot. Based on a powerful message passing scheme, HTNet requires fewer number of layers compared with other GNN-based methods which entails less supporting data and runtime. To evaluate the performance of HTNet, we prepare six different setups with more than five thousands dense dynamic WLAN deployments that cover a wide range of real-world scenarios. HTNet achieves the lowest prediction error on all six setups with an average improvement of 25.3\% over the state-of-the-art methods. 
\end{abstract}

\begin{IEEEkeywords}
throughput prediction, WLANs, machine learning, temporal graph neural network, channel bounding
\end{IEEEkeywords}

\section{Introduction}

Ever since the release of the first version of the IEEE 802.11 (Wi-Fi) standard in 1997, Wi-Fi has become the most established and most widely used Wireless Local Area Network (WLAN) technical standards. 
Recently, along with the development of Internet of Things (IoT)~\cite{7756279} and wireless smart devices, the density of WLAN deployments has dramatically increased, leading to strong interference that could seriously affect performance if not properly handled. 
To fulfill the growing requirements on latency and throughput of these wireless devices, the 802.11 standard is constantly being updated to introduce advanced features such as Dynamic Channel Bounding (DCB)~\cite{2019dcb} and multi-user Orthogonal Frequency Division Multiple Access (OFDMA)~\cite{7422404}.

A fundamental problem in WLAN optimizations is the throughput prediction problem~\cite{henty2001throughput}. An accurate and fast throughput predictor allows network managers (whether automated or not) to explore more of the design space and design high performance and robust WLAN systems. However, predicting the throughput in next generation WLANs is challenging due to the strong interference in dense deployments and the vast variety of dynamically configured WLAN features. On the one hand, classic theory-based methods that use channel statistics such as Signal to Interference and Noise Ratio (SINR) and Received Signal Strength Indicator (RSSI) fail to capture sophisticated interactions at the Media Access Control (MAC) and Physical (PHY) layers, leading to fallacious predictions~\cite{7338298}. On the other hand, event-based simulators~\cite{barrachina2019komondor,ns3} require order-of-magnitude more simulation time compared with the actual transmission time, and hence can only be used for verification purposes. In order to quickly and rapidly predict throughput, the predictor needs to analyze the multi-domain information from WLAN deployments including channel state information, connectivity information, and their trends with respect to time.

Recently, Graph Neural Networks (GNN)~\cite{kipf2017semi} and variants have achieved significant successes in learning non-Euclidean graph data. By iteratively gathering and aggregating information from neighbors, the node embeddings generated by GNNs remarkably outperform traditional graph analytic algorithms in many downstream tasks such as node classification and link prediction. However, GNNs suffer from neighbor explosion~\cite{graphsaint-iclr20} and over-smoothing~\cite{chen2020measuring} and they need to be adapted to different operating environments. In the networking domain, GNNs have demonstrated superior performance in extracting node- and graph-level features from network graphs and are widely used in various applications such as resource allocation~\cite{shen2020graph,8815526} and traffic prediction~\cite{9158437,9299774}. However, for throughput prediction, static GNNs fail to capture the important temporal information, which leads to lower accuracy. 

To design fast and accurate throughput predictors, recent works~\cite{7338298,ATARI,cbpredict} propose Machine Learning (ML) algorithms to extract features from WLAN deployments and solve the throughput prediction problem as a regression problem. ATARI~\cite{ATARI} goes one step further to directly apply GNN to extract features from WLAN deployment graphs.   However, they lack the ability to capture complete information from WLAN deployments and analyze it in a comprehensive way. We explicitly divide the overall information of a WLAN deployment into three parts:
\begin{itemize}
    \item \textbf{Contextual information} designates the self-contained static information of an AP or STA, including the primary channel, available channels, transmission power, traffic, and location.
    \item \textbf{Structural information} designates the static information related to two entities. For example, the structural information of an STA and its corresponding AP includes RSSI, SINR, and airtime.
    \item \textbf{Temporal information} designates the dynamic information that changes with time. Depending on different WLAN deployments, the temporal information could be the dynamic channel allocation, dynamic STA location, dynamic interference source, etc.
\end{itemize}
To jointly learn these three parts, we propose HTNet, a specialized Heterogeneous Temporal Graph Neural Network (HTGNN) that extracts features from WLAN deployment without any information loss. HTNet captures contextual and structural information through an attention-based heterogeneous message passing scheme. We prove that HTNet achieves maximal expressive power and can distinguish between any two different WLAN deployments. For temporal information, HTNet views WLAN deployments as Discrete Time Dynamic Graphs (DTDGs) and employs a sequence model to regulate the message passing process. Note that HTNet is designed to extract features from WLAN deployments and can be applied not only to the throughput prediction problem but also to many other network problems such as dynamic channel allocation and power control. We summarize the main contributions of this work below:
\begin{itemize}
    \item We propose HTNet, a low-complexity fast inference %
    HTGNN-based feature extractor for WLAN deployment which is the first model that jointly captures the complete contextual, structural, and temporal information in WLAN deployments.
    \item We analyze the characteristics of WLAN deployment snapshots and design a message passing scheme with maximal expressive power.
    \item We generate and publish the first dynamic WLAN deployment throughput prediction dataset that contains six different setups which cover a wide range of real-world scenarios and can be  used by other researchers.
    \item We evaluate the performance of HTNet on the throughput prediction problem on the aforementioned dataset. HTNet achieves the lowest prediction error on all six setups, outperforming state-of-the-art methods by an average of 25.3\%.
\end{itemize}
\section{Background}

\subsection{Notation}

We represent scalars as lowercase non-bold letters such as $t$ and $\alpha$, vectors as lowercase bold letters such as $\mathbf{h}$ and $\mathbf{b}$, and matrices as uppercase bold letters such as $\mathbf{A}$ and $\mathbf{W}$. Nodes are represented using the letters $u,v,i,j$ while edges are represented using the letter $e$ or their source and destination nodes $uv$ and $ij$. The variables of the $k$-th layer are denoted by $(k)$ in brackets.

\subsection{Problem Definition}
\label{sec:backpd}

We first define the throughput prediction problem in dynamic WLAN deployments. Given a dynamic WLAN deployment with some AP and STA nodes where each STA node is connected to one AP node at any time, the throughput prediction problem aims at predicting the dynamic throughput of each STA. 
We assume the time granularity is $t_g$ (i.e., the WLAN deployment changes every $t_g$ time intervals). Without loss of generality, the throughput prediction problem could be represented as a node regression problem on a heterogeneous Discrete Time Dynamic Graph (DTDG), which can be represented as a series of heterogeneous graph snapshots $\{\mathcal G_t|t=0,t_g,2t_g,\cdots\}$ in which each static snapshot $\mathcal G_t(\{\mathcal V_{it}\},\{\mathcal E_{jt}\})$ has two types of nodes AP and STA and two types of edges AP-STA and AP-AP. For simplicity, each node or each edge, despite the heterogeneity, is associated with a fixed length node or edge feature vector $\mathbf v$ or $\mathbf e$ that represents the channel, transmission, interference, and other information. The missing features are denoted as zeros in the vectors. This allows the heterogeneous DTDG to be easily converted to homogeneous DTDG to support homogeneous baseline methods. 
Denoting the throughput of STA node $i$ and time $t$ to be $y_{it}$, the throughput prediction problem is defined as a dynamic node regression problem on $\{\mathcal G_t\}$ to predict $y_{it}$ for all STA nodes $i$ at each timestamp $t=0,t_g,2t_g,\cdots$.
We adopt the supervised learning setup in which the ground truth throughput is given for the training and validation sets and the goal is to predict the throughput in the test set. Please refer to Section \ref{sec:exp} for details.

\subsection{Static Graph Neural Networks}

Static Graph Neural Networks (GNNs) encode graphs into node embeddings by iteratively performing message passing in each Graph Neural Layer (GNL). The vanilla homogeneous GNN Graph Convolutions Network (GCN) \cite{kipf2017semi} is defined on a static homogeneous undirected graph $\mathcal G(\mathcal V,\mathcal E)$ where each node $v\in\mathcal V$ is associated with some attribute features represented by a fixed length vector. The $k$-th GCN layer computes the hidden features by the following message passing scheme
\begin{equation}
    \mathbf H_v^{(k)}=\text{ReLU}(\tilde{\mathbf A}\mathbf H_v^{(k-1)}\mathbf W^{(k)}),
\end{equation}
where $\mathbf H_v$ is the hidden node feature matrix of all the nodes in the graph stacked vertically, $\tilde{\mathbf A}$ is the normalized adjacency matrix, and $\mathbf{W}$ is the weight matrix. Multiple GCN layers are stacked to compute the final output node embeddings
$\mathbf H_v^{(K)}$ which are sent to downstream applications for further processing. In the supervised training setup, the downstream networks directly provide the gradients so that the weight matrices $\mathbf W$ can be learned. The receptive field of a node is defined as the set of supporting nodes used to compute its node embedding. For example, a 2-layer GCN has the receptive field of 2-hop neighbors.
To improve performance and support various applications, later works propose many GNL variants based on the vanilla GCN layer. EGNN~\cite{gong2019exploiting} enhances GCN by maintaining additional hidden edge features $\mathbf H_e^{(k)}$ in the message passing scheme. GAT~\cite{velivckovic2018graph} introduces the attention feature when aggregating neighbor information. JK-Net~\cite{xu2018representation} improves the performance of deep GNNs with residual connections. 

A static heterogeneous graph could be represented by $\mathcal G(\{\mathcal V_i\},\{\mathcal E_j\})$ where each $i$ represent one type of node and each $j$ represent one type of edge. The edge types are also referred to as the relations $\mathcal{R}=\{j\}$. To model the heterogeneity, R\nobreakdash-GCN~\cite{schlichtkrull2018modeling} treats each relation $j$ independently and uses a combine function $\gamma(\cdot)$ to generate the output node embedding, which can be implemented as mean, concatenation, etc. A R\nobreakdash-GCN layer has the forward model
\begin{equation}
    \mathbf H_v^{(k)}=\gamma\left(\text{GNL}\left(\tilde{\mathbf A}^j,\mathbf{H}_v^{(k-1)},\mathbf W^{(k)}_j\right),j\in\mathcal R\right),
\end{equation}
where GNL is any homogeneous Graph Neural Layer, $\tilde{\mathbf A}^j$ is the normalized adjacency matrix of relation $j$, and $\mathbf{W}^{(k)}_j$ is the weights in GNL of relation $j$.

\subsection{Temporal Graph Neural Networks on DTDGs}

Dynamic graphs can be expressed as set of graph events. For example, in a dynamic WLAN deployment, the mobility of STAs can be described as a graph event in which a node feature vector changes. A hand over of an STA between two APs can be described as two graph events of an edge disappearing and another edge appearing.
In a dynamic graph, each graph event is associated with a timestamp. If the timestamps are discrete, we define these dynamic graphs as Discrete Time Dynamic Graph (DTDG). A DTDG with discrete timestamps $\{t_i\}$ can be represented by a sequence of static graph snapshots $\{\mathcal G_{t_i}\}$ in which each static graph $\mathcal G_{t_i}$ is the snap of the dynamic graph immediately after the graph events at $t_i$. To model DTDGs, static GNNs which operate in each graph snapshot are combined with sequence models to generate dynamic node embeddings at each snapshot. EvolveGCN~\cite{EvolveGCN} applies GRU or LSTM to regulate the weight matrices in each static graph snapshot. DySAT~\cite{sankar2018dynamic} applies self-attention directly on the hidden features of each snapshot.

\subsection{Related Works in Throughput Prediction}

Throughput prediction~\cite{5378489,8486372,9484520} is a classic and important problem in the networking domain. Traditional analytic algorithms rely on Channel State Information (CSI) such as SINR and RSSI to directly compute the channel throughput, which is not sufficient to predict the throughput of current dynamic and dense WLAN deployments. In fact, they even fail in a simple two links system when the CSMA/CA feature is enabled~\cite{7338298}. Owing to the drawbacks of traditional analytic algorithms, researchers have directed attention to applying ML methods as black boxes that directly map the WLAN deployments to throughput. The works~\cite{7338298,KHAN2020102499} evaluated the performance of popular ML methods including Support Vector Regression, Random Forest, and Decision Tree. While these methods outperform traditional analytic algorithms, they only exploit partial contextual information in WLAN deployments and have unsatisfactory accuracy. In 2021, the International Telecommunication Union (ITU) hosted an AI/ML challenge~\cite{cbpredict} to predict the throughput on dense {\it static} WLAN deployments with Dynamic Channel Bonding (DCB). ITU also published the first large-scale datasets which contains 800 synthetic dense WLAN deployments with 4-12 APs per deployment and 2-20 STAs per AP. The winning team Ramon~\cite{cbpredict} adopted a Multi-Layer Perceptron (MLP) to individually process the information in each Basic Service Set (BSS). The second place team ATARI~\cite{ATARI} proposed the first GNN-based predictor which exploits both contextual and structural information. However, neither of these existing works model the important temporal information in WLAN deployments as they only operate on independent snapshots.  As also observed in Lumos5G~\cite{10.1145/3419394.3423629} throughput in 5G networks is affected by mobility,  which can only be captured by a dynamic model. Similarly, the temporal information, such as the mobility of the STAs or interference sources and the changes in channel allocation, requires the ML model to take the time-related information into consideration.

Another line of work aims at developing discrete event simulators that simulate the low-level events of each transmitted package and estimate the throughput. State-of-the-art simulators~\cite{ns3,barrachina2019komondor} support advanced features in the latest 802.11ax standard and their simulation results are within tiny discrepancy from the real-world cases. However, these simulators need order-of-magnitude more simulation time compared with the transmission time, especially in high interference dense deployments. As a result, they are usually used only for verification purposes such as generating synthetic data to supervise ML models.
\section{Approach}

\begin{figure*}[t]
  \centering
  \input{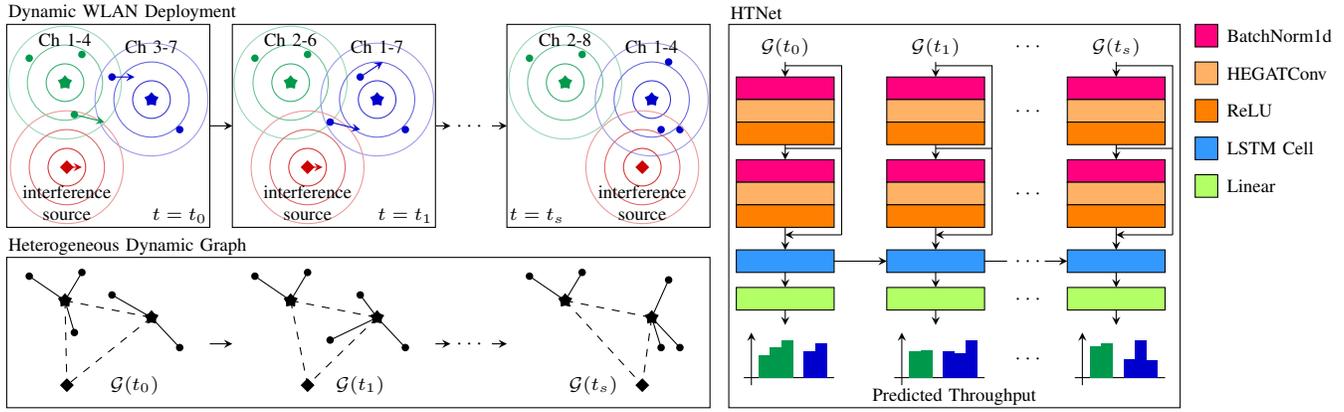}
  \caption{Overview of HTNet predicting the throughput on a dynamic WLAN deployment with two APs (star-shaped nodes), three static STAs (round-shaped nodes), two moving STAs (round-shaped nodes with arrows), and one moving interference source (diamond-shaped node). The Ch $x$-$y$ denotes that the AP is currently using channel $x$ to channel $y$. The shown neural architecture represents a 2-layer HTNet with two graph convolutional layers followed by one LSTM layer.}
  \label{fig:main}
\end{figure*}

Different GNN architectures are designed for various graph applications. To design an efficient GNN architecture on the DTDG of WLAN deployments, we first need to design the static GNN that operates on a single snapshot. The general message passing scheme of any GNN layer can be summarized using the following equation
\begin{equation}
    \bm{h}_v^{(k)}=\phi\left(\bm{h}_v^{(k-1)},\psi\left(\left\{\bm{h}_u^{(k-1)},u\in\mathcal{N}(v)\right\}\right)\right),
    \label{eq:generallayer}
\end{equation}
where $\bm{h}_v^{(k)}$ is the hidden feature of node $v$ at layer $k$, $\phi(\cdot)$ is the update function, $\psi(\cdot)$ is the aggregation function, and $\mathcal{N}(v)$ is the set of 1-hop neighbors of node $v$. To design a powerful GNN architecture on WLAN snapshots, we need to answer two important questions:
\begin{itemize}
    \item \textit{What is a good message passing scheme that allows the GNN to maximize expressive power?}
    \item \textit{How many layers does the GNN need to have?}
\end{itemize}

We arrange the rest of this section as follows. We first provide answers to these questions by analyzing the discriminative power of GNNs on WLAN graphs in Section \ref{sec:powerfulgnn}. Then, in Section \ref{sec:htnet}, we propose HTNet---a maximally powerful Heterogeneous Temporal GNN on WLAN graphs.

\subsection{Designing Powerful GNN on WLAN Graph Snapshots}
\label{sec:powerfulgnn}

Recall that any static snapshot $\mathcal{G}(t)$   consists of connected AP nodes and STA nodes, each connected only to its dedicated AP node. We first formally define the static WLAN graphs as linked stars graphs.

\begin{definition}[Star Graph]
    A star graph (star) $\mathcal{S}(\mathcal{N}, \mathcal{E})$ is an undirected graph with the single internal node $n_0\in\mathcal{N}$ connected to the rest of the nodes $\mathcal{N} \setminus n_0$ by edges $e\in\mathcal{E}$.
\end{definition}

\begin{definition}[Linked Star Graph]
    A linked star graph $\mathcal{GS}(\{\mathcal{S}_i, 0\leq i<n_s\}, \mathcal{E}_s)$ is an undirected graph with $n_s$ stars where the internal nodes $\{n_0\}$ form an undirected fully-connected graph with edges $\mathcal{E}_s$. 
\end{definition}

\begin{remark}
    Linked stars graph are a special type of heterogeneous graph $\mathcal{G}(\{\mathcal{V}_i\},\{\mathcal{E}_j\})$. For example, in a WLAN deployment, the linked star graph is a heterogeneous graph with two node types (AP and STA) and two edge types (AP-STA and AP-AP).
\end{remark}

Despite the fact that star or star-structured graphs are widely used in many applications \cite{guo-etal-2019-star,10.1145/3340531.3412014}, the expressive power of GNN on these graphs has not been analyzed. To fill this gap, we show that certain GNNs are as powerful as the Weisfeiler-Lehman (WL) test \cite{weisfeiler1968reduction} that achieves maximal expressive power on linked stars graphs by mapping any non-isomorphic pair to different embeddings.

\begin{theorem}
    \label{the:1wl}
    1-WL test can distinguish any two non-isomorphic linked stars graphs.
\end{theorem}

\begin{proof}
    We show that any pair of linked stars graphs that passes the 1-WL test is isomorphic. The first 1-WL iteration checks the degree distribution of the 1-hop neighbors. In a linked star graph, all nodes $v$ with with degree $d_v>2$ are internal star nodes and their induced subgraph of the node set $\{v, d_v>2\}$ is fully connected. Hence, for two linked star graphs with the same degree distribution, the node mapping function that maps the $d>2$ nodes according to their degree and $d=1$ nodes according to the degree of their neighbors proves that the two graphs are isomorphic.
\end{proof}

Recent works \cite{gin,maron2019provably,azizian2021expressive,10.5555/3454287.3455711} have shown that the WL test and GNN share many common properties due to their similar propagation schemes. For a GNN that follows Equation \ref{eq:generallayer} in the forward propagation, Keyulu et al. \cite[Theorem~3]{gin} prove that GNNs with injective $\phi$, $\psi$, and graph readout function are as powerful as the 1-WL test. Hence, we can conclude that GNNs with message passing schemes that meet the aforementioned injective conditions can distinguish any linked stars graphs by mapping them to different embeddings and achieve maximum expressive power. 

However, discrimination of non-isomorphic homogeneous linked star graphs is not enough in the WLAN throughput prediction problem. When the AP nodes aggregates information from their heterogeneous neighbors, the node type (AP node or STA node) can have antithetical influence on the throughput. For example, two nearby AP-STA node pairs with the same channels usually have high throughput while two nearby AP-AP node pairs leads to high interference. The aggregation function $\phi(\cdot)$ that operates on homogeneous neighbors cannot map heterogeneous multisets into injective values, even with an additional binary indicator node feature \cite{ATARI}. 
We address this issue by using an additional injective aggregation function $\gamma(\cdot)$ to model the heterogeneity. 

\noindent\textbf{GNN Depth} 
Since the maximum diameter of any linked star graph is 3 (i.e., an STA-AP-AP-STA path), the receptive field of a 3-layer GNN is large enough to cover the whole graph. Without a dedicated graph readout function (which acts as adding an extra node that is connected to all other nodes in the graph), a shallow GNN can embed the whole graph information into each node embedding. For a 3-layer GNN with injective $\phi$ and $\psi$, the output node embeddings can distinguish any differences in the graph structure or the node features. 
However, for a less powerful GNN, more layers are needed to compensate the lack of resolution in the aggregation function, as shown in shaDow-GNN \cite{shaDow} where the extra depth improves the performance of GNNs with bounded-size receptive field. In experiments, we verify that a less powerful GNN requires more number of layers to achieve its best performance compared with a maximally expressive GNN. 
We further discover that 2-layer HTNet which provides information of all APs to each STA achieves comparable accuracy as a 3-layer HTNet while requiring less runtime and less amount of supporting features, due to the fact that the initial AP features already includes a fair amount of information from their corresponding STAs. Please refer to Section \ref{sec:result} for details.

\subsection{HTNet Architecture}
\label{sec:htnet}

Figure \ref{fig:main} shows an overview of HTNet predicting the throughput on a dynamic WLAN deployment with $s+1$ snapshots. We first introduce HTLayer (HTL) which is the basic building block for HTNet. For reasons of expediency, we omit the time $t=t_0,t_1,\cdots,t_s$ when introducing HTL.

\noindent\textbf{HTL}: Consider a WLAN snapshot $\mathcal{G}$ %
at time $t$ with $n$ APs and $m$ STAs. We model the external interference sources (red nodes in Figure \ref{fig:main} as close AP-STA pairs that occupy all available channels. There are $n(n-1)/2$ AP-AP edges modeling the interference among APs and $m$ STA-AP edges modeling the transmission between STAs and APs. In the message passing process, an undirected edge is treated as two directed edges with opposite directions. Hence, we have three types of directed edges: $n(n-1)$ AP$\to$AP edges, $m$ STA$\to$AP edges, and $m$ AP$\to$STA edges. Let the node feature matrix be $\{\mathbf{h}_v^{(0)}\}=\mathbf{H}_v^{(0)}\in\mathbb{R}^{n+m,d_v}$ and the edge feature matrix be $\{\mathbf{h}_{uv}^{(0)}\}=\mathbf{H}_e^{(0)}\in\mathbb{R}^{n(n-1)+2m,d_e}$ where $d_v$ and $d_e$ are the lengths of node and edge feature vectors. Note that the directed edges copied from the undirected edges in the message passing process share the same edge feature matrix. 

For each relation $r\in\{\text{AP}\to\text{AP},\text{STA}\to\text{AP},\text{AP}\to\text{STA}\}$, we need to find an injective aggregation function $\psi(\cdot)$ that operates on homogeneous neighbor multisets. The injective combine function $\phi(\cdot)$ is replaced with a residual connection and will be introduced later. Under the assumption of countable feature space, GIN~\cite{gin} shows that the sum aggregator is a universal injective function over homogeneous multisets. To accommodate the edge features, we enhance the sum aggregator with edge-aware heterogeneous attention mechanism HEGATConv~\cite{kaminski2022rossmann}. For the $k$-th HTL layer with input and output node and edge dimensions $d_v^\text{in(k)},d_v^\text{out(k)},d_e^\text{in(k)},d_e^\text{out(k)}$, the hidden node features $\mathbf{h}_v^{(k)}(r)$ of relation $r$ is updated by
\begin{equation}
    \mathbf{h}_v^{(k)}(r)=\text{ReLU}\left(\sum_{u\in\mathcal{N}_r(v)}a_{uv}\mathbf{W}_r^{(k)}\mathbf{h}_u^{(k-1)}+\mathbf{b}_r^{(k)}\right),
\end{equation}
where $\mathcal{N}_r(v)$ is the neighbor set of node $v$ with relation $r$, $\mathbf{W}_r^{(k)}\in\mathbb{R}^{d_v^\text{out(k)},d_v^\text{in(k)}}$ is the learnable weight matrix of relation $r$, $\mathbf{b}_r^{(k)}\in\mathbb{R}^{d_v^\text{out(k)}}$ is the learnable bias matrix of relation $r$. And the edge feature-aware attention scores $a_{uv}$ are computed by
\begin{align}
    a_{uv}&=\mathbf{w}_a^{(k)}\mathbf{h}_{uv}^{(k)}\\
    \mathbf{h}_{uv}^{(k)}&=\text{LeakyReLU}\left(\mathbf{W}_a^{(k)}\left[\mathbf{h}_u^{(k-1)}||\mathbf{h}_{uv}^{(k-1)}||\mathbf{h}_v^{(k-1)}\right]\right),
\end{align}
where $\mathbf w_a^{(k)}\in\mathbb{R}^{d_e^\text{out(k)}}$ is the learnable attention weight vector, $\mathbf{W}_a^{(k)}\in\mathbb{R}^{d_e^\text{out(k)},2d_v^\text{in(k)}+d_e^\text{in(k)}}$ is the learnable edge weight matrix, and $||$ denotes the concatenation operation. Here the attention weight vector $\mathbf w_a^{(k)}$ and the edge weight matrix $\mathbf W_a^{(k)}$ are shared among all relations and the hidden edge features $\mathbf h_{uv}^{(k)}$ are updated for all edge types prior to the message passing process of each relation. 

For the heterogeneous combine function $\gamma(\cdot)$, since the number of relations is only three, we simply choose the most expressive concatenation function. The output node features of the $k$-th HTL layer is
\begin{equation}
    \mathbf{h}_v^{(k)}=\mathop{||}_{r\in\mathcal{R}}\mathbf{h}_v^{(k)}(r).
\end{equation}
For the STA nodes that only have one incoming AP$\to$STA edge, the missing $\mathbf{h}_v^{(k)}(r)$ for $r\in\{\text{AP}\to\text{AP},\text{STA}\to\text{AP}\}$ is replaced with all zero vectors.

\noindent\textbf{Over-Smoothing}: The over-smoothing problem usually occurs in deep GNNs when the node embeddings of all nodes converge to the same vector and lose the ability to distinguish each individual node. Although HTNet only requires 2 or 3 layers, the special graph structure of WLAN deployments exacerbates the over-smoothing problem. Any two STAs under the same AP share the same sets of 1-hop and 2-hop neighbors. Any two STAs under different APs even have a common 2-hop neighbor set of all other APs. To mitigate the over-smoothing issue, we adopt the JK\nobreakdash-Net~\cite{xu2018representation} setup to add a residual connection after each HTL. The output node embedding of $K$-layer HTLs is
\begin{equation}
    \mathbf h_v=\phi\left(\mathop{||}_{0\leq k\leq K}\mathbf{h}_v^{(k)}\right),
\end{equation}
where the injective $\phi(\cdot)$ function is implemented using a 1\nobreakdash-layer MLP.

\noindent\textbf{HTNet}: After HTLs perform message passing on each WLAN snapshot and obtain $\mathbf{h}_v^t$ at $t=t_0,t_1,\cdots,t_s$, we apply a sequence model to capture the temporal information. Among the popular sequence models, we find in experiments that LSTM achieves the best performance compared with RNN, GRU, and Transformer. Hence, we implement LSTM network to compute the final dynamic node embedding $\tilde{\mathbf{h}}_v^t$
\begin{align}
    \mathbf{f}_v^t&=\sigma\left(\mathbf{W}_f\mathbf{h}_v^t+\mathbf{U_f}\tilde{\mathbf{h}}_v^{t-1}+\mathbf{b}_f\right)\\
    \mathbf{i}_v^t&=\sigma\left(\mathbf{W}_i\mathbf{h}_v^t+\mathbf{U_i}\tilde{\mathbf{h}}_v^{t-1}+\mathbf{b}_i\right)\\
    \mathbf{o}_v^t&=\sigma\left(\mathbf{W}_o\mathbf{h}_v^t+\mathbf{U_o}\tilde{\mathbf{h}}_v^{t-1}+\mathbf{b}_o\right)\\
    \mathbf{c}_v^t&=\mathbf{f}_v^t\circ\mathbf{c}_v^{t-1}+\mathbf{i}_v^t\circ\sigma\left(\mathbf{W}_c\mathbf{h}_v^t+\mathbf{U}_c\tilde{\mathbf{h}}_v^{t-1}+\mathbf{b}_c\right)\\
    \tilde{\mathbf{h}}_v^t&=\mathbf{o}_v^t\circ\sigma\left(\mathbf{c}_v^t\right),
\end{align}
where $\sigma(\cdot)$ is the sigmoid function, $\circ$ is the point-wise multiplication operator, $\mathbf{W}$ and $\mathbf{U}$ are the learnable weight matrices, $\mathbf{b}$ is the learnable bias, $t-1$ denotes the previous timestamp, and $\mathbf{c}_v^t$ is initialized to be $\mathbf{0}$. We apply a linear layer to generate the predicted throughput $\hat{y}_v^t$ and add a softplus function to ensure that the predicted throughput is strictly positive
\begin{equation}
    \hat{y}_v^t=\text{log}\left(1+\text{exp}\left(\mathbf{W}_y\tilde{\mathbf{h}}_v^t\right)\right).
\end{equation}
We apply batch normalization with learnable affine parameters after each HTL. All the learnable parameters in HTNet are trained end-to-end supervised by the ground truth throughput $y_v^t$ in the training set. The loss function is to minimize the Root Mean Square Error (RMSE) of all target STAs across all the time
\begin{equation}
    \mathcal{L}=\sqrt{\frac{\sum_v\sum_t\left(y_v^t-\hat{y}_v^t\right)^2}{|\{v\}||\{t\}|}}.
\end{equation}

\noindent\textbf{Comparison with ATARI}: HTNet and ATARI~\cite{ATARI} both apply GNN to extract features from graph-structured WLAN data and predict the throughput. However, HTNet considers the heterogeneity between AP and STA nodes and adopts attention mechanism in aggregation, which are proven to achieve maximized expressiveness on WLAN graphs. In addition, HTNet also captures the temporal information jointly with the structural and contextual information in WLAN graphs.

\noindent\textbf{HTNet Applications}: HTNet can be widely applied to many WLAN problems. In the first place, as a throughput predictor, HTNet can be used to determine the reward in Reinforcement Learning (RL) algorithms~\cite{8993737,s20102789} to learn channel allocation policies. More generally, the dynamic node embeddings $\tilde{\mathbf{h}}_v^t$ generated by HTNet, by virtue of their high expressiveness in low-dimensional spaces, can be directly used in downstream applications, such as channel allocation, link scheduling, power control, beamforming, and network flow optimization.
\section{Experiments}
\label{sec:exp}

\begin{figure*}[t]
  \centering
  \input{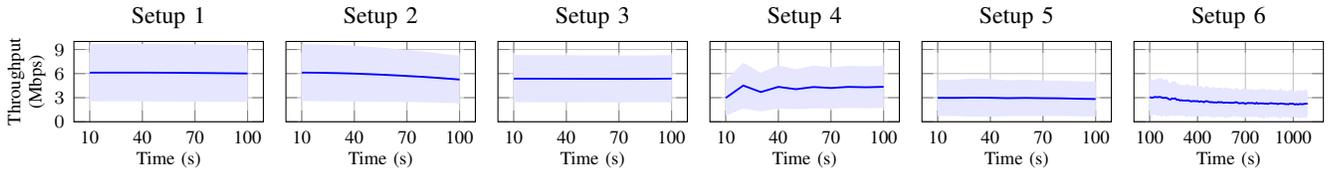}
  \caption{Average dynamic throughput (Mbps) in each setup. The shaded regions show the standard deviations in each snapshot.}
  \label{fig:dsmean}
\end{figure*}

We arrange the section as follows. We first introduce the dataset in Section \ref{sec:ds}. Then, we introduce the baseline methods in Section \ref{sec:bs}. In Section \ref{sec:code}, we discuss the implementation of HTNet and the baselines. Lastly, we present the results and ablation study in Section \ref{sec:result}. The codes and dataset are publicly available at GitHub\footnote{\url{https://github.com/tedzhouhk/HTNet}}.

\subsection{Dataset}
\label{sec:ds}

HTNet is the first work that considers the dynamics in WLAN deployments and to the best of our knowledge, there are no existing WLAN deployment datasets that include temporal information. Fortunately, discrete event simulators can be used to generate large-scale synthetic data for the ML models. We choose Komondor~\cite{barrachina2019komondor}, the fastest discrete event simulator that supports 802.11ax features, which has been cross-validated with ns-3~\cite{ns3} and real-world test-benches. Note that HTNet is a supervised model and its performance can be further improved through more accurate training data such as real-world WLAN deployments. There are three types of dynamics in WLAN deployments: (1) mobility of STAs. For example, a customer streams music using a smart phone while walking in a department store. (2) dynamic interference sources, e.g., an external Device-to-Device (D2D) network with dynamic channel and power configurations. In this work, we only focus on mobile dynamic interference sources with fixed power and channel configurations. (3) dynamic channel allocation, i.e., the available channels are dynamically configured by APs.

\begin{table}[b]
    \centering
    \setlength{\tabcolsep}{5pt}
    \begin{tabular}{r|cccc}
         & \makecell{Sequence\\Length} & \makecell{Mean Thpt.\\(Mbps)} & \makecell{STD Thpt.\\(Mbps)} & \makecell{\# Deployments\\Train/val/test}\\
        \toprule
        Setup 1 & 10 & 6.02 & 3.56 & 3000/1000/1000\\
        Setup 2 & 10 & 5.82 & 3.35 & 3000/1000/1000\\
        Setup 3 & 10 & 5.35 & 2.93 & 3000/1000/1000\\
        Setup 4 & 10 & 4.11 & 2.63 & 3000/1000/1000\\
        Setup 5 & 10 & 2.94 & 2.25 & 3000/1000/1000\\
        Setup 6 & 100 & 2.48 & 1.89 & 300/100/100\\
    \end{tabular}
    \caption{Dataset statistics.}
    \label{tab:ds}
\end{table}

To cover a wide range of real-world dynamic WLAN deployments, we generate a dynamic dataset with six setups. Setups 1-2 cover the case of mobile STAs. Setup 3 covers the case of mobile interference sources. Setup 4 covers the case of dynamic channel configuration. Setups 5-6 cover more complex cases with both mobile STAs and dynamic channel configurations. To generate these dynamic setups, we use the six training scenarios in the ITU AI/ML Challenge dataset~\cite{cbpredict} as a starting point. These six training scenarios contain 500 static dense WLAN deployments with the map sizes ranging from 80 by 60 to 40 by 20 meters, number of APs per deployment ranging from 8 to 12, and number of STAs per AP ranging from 5 to 20. We adopt the same WLAN configuration as in the ITU AI/ML Challenge --- 802.11ax Wi-Fi with 8 consecutive 20MHz 5G channels and Always Max Log2 dynamic channel bonding mechanism. The traffic is set to be full-buffered downlink UDP traffic. Starting from these 500 static deployments, we generate six dynamic setups. The first four setups focus on different dynamic scenarios while the last two setups cover more general use cases.
\begin{itemize}
    \item \textbf{Setup 1: mobile STAs (within AP).} In setup1, STAs are moving along straight lines. We randomly pick 50\% of the STAs of each AP to be the mobile STAs while the remaining 50\% are static. The trajectory of these mobile STAs are straight lines pointing at random directions. Starting from the 500 initial WLAN deployments, the moving speeds are chosen uniformly randomly from 0.1-0.5 meters and are fixed in the dynamic sequences. If an STA is moving out of the coverage range of its AP, it stops moving. In setup 1, 16\% of the total movements result in out of coverage and are stopped.
    \item \textbf{Setup 2: mobile STAs (across AP).} Setup 2 is similar to setup 1 where there are 50\% of mobile STAs moving along random directions. However, in setup 2, we double the upper bound of their moving speeds, ranging from 0.1-1 meters per snapshot. In addition, if one STA is moving out of the coverage range of its AP, it is handed-over to the closest AP that can cover the STA. The STA will stop moving only if it moves out of the coverage range of all APs in the deployment. In setup 2, 6\% of the total movements result in out of coverage and are stopped.
    \item \textbf{Setup 3: mobile interference sources.} Setup 3 covers the cases of dynamic interference sources. As mentioned above, in this work, we only consider mobile interference source with fixed power and channel. Specifically, we simulate an external interference source as a AP-STA pair capturing all available 8 channels placed very close to each other (the STA is placed 1cm above the AP.) with the largest transmission power available. In setup 3, we randomly place 3 interference sources in each  deployment. These interference sources are moving along random directions with constant speeds chosen uniformly from 1-3 meters per snapshot at random.
    \item \textbf{Setup 4: dynamic channel allocation.} Setup 4 covers the cases where APs are actively changing the available channels during transmission. 
    We adopt a random channel allocation scheme: In each snapshot, the APs have equal probability to (1) increase minimum available channel number, (2) decrease minimum available channel number, (3) increase maximum available channel number, (4) decrease maximum available channel, or (5) add an offset to both the minimum or maximum available channel numbers.
    \item \textbf{Setup 5: mobile STAs + dynamic channel allocation.} Setup 5 is a combination of setups 2 and 4.
    \item \textbf{Setup 6: long sequence.} Setup 6 is also a combination of setups 2 and 4, but with longer sequence.
\end{itemize}
The sequence length is 10 snapshots for setups 1-5 and 100 snapshots for setup 6 where each snapshot represents 10 seconds in the real world. 
Random movement directions are chosen for the first snapshot and remain fixed for the remaining snapshots.
For setups 1-5, we generate 10 different dynamic WLAN deployments from one static WLAN deployment, while for setup 6 we only generate 1. The dataset statistics are shown in Table \ref{tab:ds} and the average throughput in each snapshot is shown in Figure \ref{fig:dsmean}. We use a server with dual AMD EPYC 7763 CPU and 1TB of DDR4 memory to run the simulation. We launch 128 independent Komondor simulation processes at the same time to fully utilize the 128 cores. The whole simulation process takes around one week.

\begin{table*}[ht!]
  \centering
  \caption{Test RMSE and MAE (Mbps) on the six experiment setups. The \ding{51} or (\ding{51}) mark denotes that the corresponding method uses full or partial contextual, structural, and temporal information. (\textbf{First} \underline{second})}
  \setlength{\tabcolsep}{3pt}
  \begin{tabular}{r|c|c|c||c||cc|cc|cc|cc|cc|cc}
     & \multirow{2}{*}{\rotatebox[origin=c]{90}{\parbox[c]{.6cm}{\linespread{0.6}\selectfont\centering Cont-extual}}} & \multirow{2}{*}{\rotatebox[origin=c]{90}{\parbox[c]{.6cm}{\linespread{0.6}\selectfont\centering Stru-ctural}}} & \multirow{2}{*}{\rotatebox[origin=c]{90}{\parbox[c]{.6cm}{\linespread{0.6}\selectfont\centering Tem-poral}}} & \multirow{2}{*}{\# Param.} & \multicolumn{2}{c|}{Setup 1} & \multicolumn{2}{c|}{Setup 2} & \multicolumn{2}{c|}{Setup 3} & \multicolumn{2}{c|}{Setup 4} & \multicolumn{2}{c|}{Setup 5} & \multicolumn{2}{c}{Setup 6}\\
     & & & & & RMSE & MAE & RMSE & MAE & RMSE & MAE & RMSE & MAE & RMSE & MAE & RMSE & MAE\\
    \toprule
    SINR & (\ding{51}) & & & - & 5.3280 & 3.7653 & 5.2527 & 3.6073 & 4.8163 & 3.3863 & 4.0335 & 2.4572 & 3.4028 & 2.0480 & 2.7934 & 1.7096\\
    GBRT & \ding{51} & & & - & 4.4906 & 3.0936 & 4.5773 & 3.0761 & 4.1263 & 2.8456 & 3.2776 & 1.9163 & 2.7991 & 1.5650 & 2.2748 & 1.2934\\
    ATARI & \ding{51} & (\ding{51}) & & 75,315 & 4.3630 & 2.9407 & 4.8505 & 3.2958 & 4.3230 & 3.0650 & 3.8030 & 2.3475 & 3.1699 & 1.8319 & 2.5271 & 1.4876\\
    Ramon & \ding{51} & & & 20,275 & 2.3928 & 1.6809 & 2.9664 & 2.0042 & 2.6728 & 1.7966 & \underline{2.7313} & \underline{1.5415} & \underline{2.4016} & \underline{1.2884} & 2.2707 & \underline{1.1688}\\
    ATARI+LSTM & \ding{51} & (\ding{51}) & \ding{51} & 339,507 & 3.1443 & 2.1183 & 3.4640 & 2.2904 & 3.4249 & 2.3683 & 3.3683 & 1.9282 & 2.7456 & 1.4897 & 2.2358 & 1.2359\\
    Ramon+LSTM & \ding{51} & & \ding{51} & 284,467 & \underline{1.5849} & \underline{1.0501} & \underline{2.6594} & \underline{1.7072} & \underline{2.3313} & \underline{1.4949} & 2.7429 & 1.5519 & 2.4693 & 1.2944 & \underline{2.2175} & 1.1778\\
    \midrule
    HTNet & \ding{51} & \ding{51} & \ding{51} & 525,619 & \textbf{1.4705} & \textbf{0.9496} & \textbf{1.9040} & \textbf{1.1976} & \textbf{1.7928} & \textbf{1.1632} & \textbf{1.9760} & \textbf{1.0441} & \textbf{1.7102} & \textbf{0.8627} & \textbf{1.5559} & \textbf{0.7963}
  \end{tabular}
  \label{tab:acc}
\end{table*}

\begin{figure*}[t]
  \centering
  \input{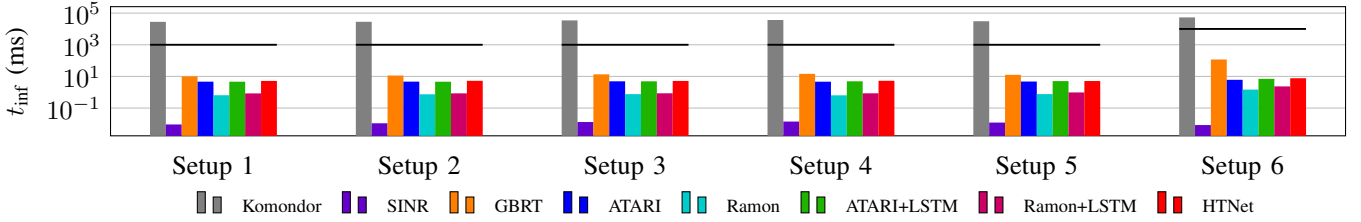}
  \caption{Average inference time per sequence compared with the simulation time of the Komondor simulator. The solid lines denote the actual time of the dynamic deployment.}
  \label{fig:runtime}
\end{figure*}

\subsection{Baseline Methods}
\label{sec:bs}

To compare the performance of HTNet with existing works, we choose four baselines.
\begin{itemize}
    \item \textbf{SINR} is based on the theoretical single link channel capacity equation $c=\Gamma\log(1+\text{SINR})$.
    \item \textbf{GBRT}~\cite{7338298} uses Gradient Boosted Regression Trees which predict the throughput through multiple independent regression trees. We choose GBRT as it achieves the best performance among the out-of-the-box ML methods~\cite{7338298,KHAN2020102499}.
    \item \textbf{ATARI}~\cite{ATARI} is the second place winner in the ITU AI/ML challenge. ATARI applies GNN to solve the WLAN throughput prediction problem.
    \item \textbf{Ramon}~\cite{cbpredict} is the first place winner in the ITU AI/ML challenge. Ramon uses a feed-forward deep learning algorithm to predict throughput from signal quality and AP bandwidth.
\end{itemize}
We also extend ATARI and Ramon to capture temporal information by adding an additional LSTM model (\textbf{ATARI+LSTM} and \textbf{Ramon+LSTM}). The number of parameters in each model is shown in Table \ref{tab:acc}.

\subsection{Implementation}
\label{sec:code}

We predict the current throughput of the target STAs given the current and all previous snapshots. The ground truth throughput of previous snapshots is not used as input data as it is not available in most real-world use cases. However, HTNet can be easily adapted to different problem setups such as the extrapolation setup which predicts future throughput without the future snapshots.
To ensure a fair comparison, we make sure that HTNet and all baseline methods share the same set of input features. We first define the node and edge features for the graph-based methods ATARI, ATARI+LSTM, and HTNet. As mentioned in Section \ref{sec:backpd}, the heterogeneous nodes and edges have the same feature vectors to be compatible with homogeneous GNNs. In this work, the input node features are 21-dimensional floating point vectors including the following information:
\begin{itemize}
    \item Node type (1-dimensional): 0 for AP, 1 for STA.
    \item Position (2-dimensional): $x$ and $y$ position of the AP or STA.
    \item Primary channel (8-dimensional): one-hot vector indicating the primary channel.
    \item Available channels (8-dimensional): multi-hot vector indicating the available channels.
    \item Airtime (1-dimensional): airtime for AP, 0 for STA.
    \item SINR (1-dimensional): SINR for STA, 0 for AP.
\end{itemize}
The input edge features are 4-dimensional floating point vectors including the following information:
\begin{itemize}
    \item Edge type (1-dimensional): 0 for AP-STA, 1 for AP-AP.
    \item Distance (1-dimensional): physical distance between two nodes.
    \item RSSI (1-dimensional): RSSI for AP-STA, 0 for AP-AP.
    \item Interference (1-dimensional): interference for AP-AP, 0 for AP-STA.
\end{itemize}
The airtime, SINR, RSSI, and interference are generated by Komondor in the simulation process. Note that we use a one-hot vector for primary channel and a multi-hot vector for available channels instead of scalars as in ATARI~\cite{ATARI}, since the channels are independent of each other. For the non-graph methods (GBRT, Ramon, Ramon+LSTM), the input features are the concatenation of the node features and the edge features of the corresponding AP-STA edges. 
All experiments are performed on dual AMD EPYC 7763 CPU with 1TB of DDR4 memory and an RTX A6000 GPU with 48GB of GDDR6 memory. The SINR and GBRT baselines are implemented on CPU while the other methods including HTNet are implemented on GPU. Please refer to our open-sourced codes for more implementation details.

\begin{figure*}[t]
  \centering
  \input{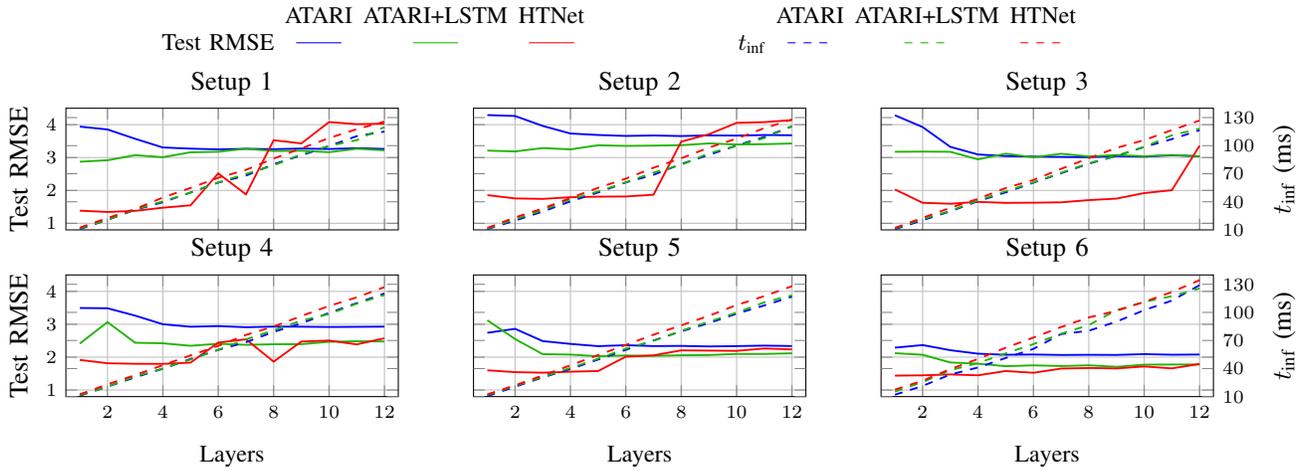}
  \caption{Test RMSE and inference time of ATARI, ATARI+LSTM, and HTNet with different number of GNN layers.}
  \label{fig:layer}
\end{figure*}

\subsection{Results}
\label{sec:result}

\noindent\textbf{Accuracy}: To compare the accuracy of HTNet with the baselines, we set the number of layers $K$ to be 2 and hidden feature dimensions $d_v^\text{out(k)},d_e^\text{out(k)}$ to be 128 for the neural network-based methods (ATARI, Ramon, ATARI+LSTM, Ramon+LSTM, and HTNet). The LSTM model in ATARI+LSTM, Ramon+LSTM, and HTNet also has two layers with 128 hidden dimension. We train these methods with 0.001 as the learning rate and 0 as the dropout rate for 150 epochs until convergence. For GBRT, we uses 100 gradient boosted trees with maximum depth of 4. Table \ref{tab:acc} shows the Root Mean Square Error (RMSE) and Mean Absolute Error (MAE) of HTNet and the baselines on the six setups. We also show the type of input information considered by these methods. HTNet is the only method that considers the complete contextual, structural, and temporal information. ATARI operates on homogeneous graphs which only captures partial structural information while SINR only captures partial contextual information. On all six setups, HTNet achieves the lowest prediction error with an average of 1.7345 RMSE and 0.3721 MAE improvement compared with the most accurate baselines. Ramon+LSTM achieves the lowest RMSE among the baselines on setups 1, 2, 3, and 6 while RAMON achieves the lowest RMSE on setups 4 and 5. The baseline methods ATARI and Ramon, when attached with a LSTM network to capture the temporal information, outperform their original static version, which is solid evidence that capturing temporal information is important in the WLAN throughput prediction problem. We observe that when the channel configuration is dynamically allocated (setups 4-6), temporal information needs to be combined with structural information to generate accurate throughput prediction. On the more general setups 5 and 6, HTNet achieves significantly better accuracy than all the baseline methods due to its high generalizability and robustness. 
We also test the static version of HTNet by removing the LSTM layer and compare the result with ATARI. On six datasets, HTNet without LSTM achieves an average of 37.1\% improvement than ATARI, which verified that the heterogeneous GNN architecture in HTNet is superior than the homogeneous GNN architecture used in ATARI.

\begin{table}[b]
    \centering
    \setlength{\tabcolsep}{2.7pt}
    \begin{tabular}{r|cccccc}
         & Setup 1 & Setup 2 & Setup 3 & Setup 4 & Setup 5 & Setup 6\\
        \toprule
        HTNet & 1.4705 & 1.9040 & 1.7928 & 1.9760 & 1.7102 & 1.5559\\
        \midrule
        w/o SINR & 1.3980 & 1.9085 & 1.9215 & 2.0145 & 1.8239 & 1.5645\\
        w/o airtime & 1.5278 & 2.0154 & 1.9858 & 2.1003 & 1.8124 & 1.5784\\
        w/o RSSI & 1.4635 & 1.9045 & 2.0154 & 2.0112 & 1.9112 & 1.6125\\
        w/o channel & 1.6845 & 2.3125 & 2.1023 & 2.4156 & 1.9652 & 1.7254\\
        w/o position & 1.7141 & 2.4012 & 2.3142 & 2.3241 & 1.9532 & 1.7098\\
        \midrule
        o channel & 3.8147 & 4.2100 & 3.6845 & 3.5487 & 3.3489 & 2.9896\\
        o position & 3.6467 & 3.9963 & 4.1085 & 3.4637 & 3.0012 & 2.8711\\
    \end{tabular}
    \caption{Test RMSE of HTNet with different set of input features. `w/o' denotes without while `o' denotes only.}
    \label{tab:ablation}
\end{table}

\noindent\textbf{Runtime}: Figure \ref{fig:runtime} shows the average inference time $t_\text{inf}$ per dynamic WLAN deployment snapshot. 
The simulation time of Komondor is more than 10 times longer compared with the real-world transaction time and becomes the bottleneck when applied in WLAN optimizations problems. SINR, due to its simplicity, requires less than 0.1ms inference time on all six setups. Ramon and Ranmon+LSTM are in the second tier with less than 1ms inference time, since they do not require aggregation from neighbor nodes. Graph-based methods ATARI, ATARI+LSTM, and HTNet are in the third tier with 4-8ms inference time. Note that although the message passing scheme in HTNet is more complex than in ATARI, the runtime of HTNet is only 10.5\% longer since the most time-consuming operation on GPU is the sparse-dense matrix operation which is the same for both. GBRT that operates on CPU is the slowest method. We believe the inference time of all ML methods meet the requirements of various downstream applications.

\noindent\textbf{Number of GNN layers}: We study the accuracy and inference time of graph-based methods ATARI, ATARI+LSTM, and HTNet with different number of GNN layers. We vary the number of GNN layers from 1 to 12 and show the results in Figure \ref{fig:layer}. Since we are performing full-batch inference (predict the throughput of all STAs at the same time), the runtime is linear with the number of GNN layers for all methods. HTNet with 3 GNN layers is proven to achieve maximal expressive power and achieves the highest test RMSE. In practice, the 2-layer HTNet achieves similar test RMSE as a 3-layer version but only requires the STA nodes of the same AP and the AP nodes as the supporting nodes and has the best accuracy-to-runtime ratio. Note that the accuracy of deeper HTNets drops due to a combination effect of over-smoothing and inferior convergence, especially in simple WLAN deployments (setups 1-3). On the other hand, less powerful GNNs ATARI and ATARI+LSTM require more number of layers to compensate the lack of expressive power in their message passing schemes. Although the temporal information in ATARI+LSTM makes up the deficiency, in complex WLAN deployments (setups 5-6), it still requires more than 10 layers for its best accuracy. 

\noindent\textbf{Input features:} We also perform ablation study on the contribution of each input feature and show the results in Table \ref{tab:ablation}. We first remove the SINR, airtime, RSSI, channel configuration (including primary channel and available channels, denoted as `channel' in Table \ref{tab:ablation}), and location information (including position and distance, denoted as `position' in Table \ref{tab:ablation}) one at a time from the input features. Removing SINR from the input features only has mild affect on the accuracy, except on setup 3 with mobile interference source. Channel configuration and location are important input features, and channel configuration is the most important input feature on setups 4-6 with dynamic channel allocation. 

\section{Conclusion}

We proposed HTNet -- the first HTGNN-based dynamic WLAN performance predictor. We designed a special message passing scheme and proved that it achieved maximal expressive power on dynamic WLAN deployment graphs. We generated the first dynamic WLAN throughput prediction dataset which contains more than 25 thousand dynamic WLAN deployments. HTNet outperformed state-of-the-art methods in all six experimental setups. We believe that the high quality dynamic embeddings and the accurate throughput prediction in HTNet can benefit many important network problems.

\section*{Acknowledgment}
This work is supported by the National Science Foundation (NSF) under grant OAC-2209563 and the DEVCOM Army Research Lab (ARL) under grant W911NF2220159. 

\nocite{*}
\bibliographystyle{IEEEtran}
\bibliography{sample}

\end{document}